\PassOptionsToPackage{dvipsnames}{xcolor}	
\documentclass[12pt]{article}

\usepackage[T1]{fontenc}
\usepackage{comment}
\usepackage{lmodern}
\usepackage{setspace}
\usepackage{etoolbox}
\makeatletter
\patchcmd{\appendix}{\@Alph}{\@Roman}{}{}
\makeatother

\usepackage{amsmath}
\usepackage{amssymb}
\usepackage{amsthm} 
\usepackage{mathtools}
\usepackage{mathrsfs}
\usepackage{breqn}
\usepackage{bigints}
\usepackage{bbm}

\usepackage[margin=1in]{geometry}
\usepackage{enumerate}
\usepackage{enumitem}
\setlist[enumerate,1]{label=(\arabic*)}
\setlist[itemize,1]{label=--}    
\usepackage{xcolor}
\usepackage{hyperref}
\usepackage{float}
\usepackage{indentfirst}
\usepackage{caption}

\usepackage{sectsty}
\sectionfont{\centering}
\subsectionfont{\centering}

\usepackage{tikz}
\usetikzlibrary{decorations.pathreplacing}
\usepackage{mathtools}
\usetikzlibrary{positioning}
\usepackage{graphicx}
\usepackage{pgfplots}
\pgfplotsset{compat=1.18} 
\definecolor{dodgerblue}{rgb}{0.1211,0.5664,1}
\definecolor{golden}{RGB}{255, 206, 84}
\definecolor{chocolate}{rgb}{1,0.5,0.14}
\definecolor{burgundy}{rgb}{0.78125, 0.1, 0.246}
\definecolor{budgreen}{rgb}{0.53125, 0.7226, 0.3125}
\definecolor{almond}{rgb}{0.985, 0.922, 0.785}

\newcommand{\bb}{\mathbb}

\newcommand{\lims}{\lim\limits}

\newcommand{\mcal}{\mathcal}

\renewcommand{\epsilon}{\varepsilon}

\newtheorem{theorem}{Theorem}

\newtheorem{lemma}{Lemma}

\newtheorem{proposition}{Proposition}
\newtheorem{observation}{Observation}

\newtheorem*{statement*}{Result}

\theoremstyle{definition}
\newtheorem{definition}{Definition}

\DeclareMathOperator*{\argmax}{\arg\!\max}

\DeclareTextFontCommand{\emph}{\slshape}

\usepackage{natbib}
\nocite{*}

\title{Fair Commodity Taxation}
\author{Eric Gao and Daniel Luo\thanks{Massachusetts Institute of Technology, ericgao@mit.edu, daniel57@mit.edu. We are indebted to Drew Fudenberg and Alex Wolitzky for invaluable guidance and Eric Yan for conversations early in the development of this project. We also thank Piotr Dworczak, Marina Halac, Michelle Hyun-Kim, Navin Kartik, Andrew Komo, Elliot Lipnowski, Stephen Morris, Ellen Muir, Bryant Xia, Ivan Werning, Kai Hao Yang, and seminar and conference participants at MIT, Yale, NEEPC 2026, Stony Brook 2026, and IMD Days for helpful comments. Coarse.ink and Refine.ink were used to check the paper for clarity and correctness. Luo acknowledges financial support from
the NSF Graduate Research Fellowship. Part of this paper was written while Luo was visiting Yale University. All errors are, of course, ours alone.}}
\date{\today}

\begin{document}
\maketitle
\begin{abstract}
We study optimal taxation in monopoly screening problems with a redistribution-minded regulator. When the regulator moves before the monopolist, any nondecreasing allocation can be implemented as the monopolist's optimum under an appropriate tax schedule. Despite this flexibility, we show that our redistribution-minded regulator never strictly benefits from randomization or nonlinear taxes. Under a strengthening of Myersonian regularity, which we call \emph{strong regularity}, we characterize the tax rates on the fairness-efficiency frontier: it is enough for the regulator to consider policies which consist of a per-unit excise tax and a lump-sum rebate, with all (and only) tax rates between the consumer-surplus maximizing rate and the Laffer peak appearing on the frontier. Finally, we show subsidies are never used, and that every frontier threshold policy rations more than the unregulated monopolist. 
\end{abstract}


\newpage
\onehalfspacing

\section{Introduction}

Commodity taxes are often a staple of redistributive taxation. For example, luxury taxes are frequently implemented by U.S. state governments---such as Connecticut's tax on jewelry and handbags (\cite{ConnecticutDRSUseTax}) and Washington's tax on luxury motor vehicles (\cite{WashingtonDORLuxuryVehicleTax2025})---and were utilized by the US federal government as recently as the 1990s. Moreover, commodity taxes are often used in lieu of income taxes by developing countries as their primary redistributive instrument (\cite{besley2014taxing}). 

The redistributive impact of these taxes is often at the center of public debate. For example, opponents of California’s excise tax on gas criticize its supposedly regressive nature, while proponents argue that the burden falls primarily on more affluent households, who already benefit disproportionately from market activity. Despite their prevalence, a formal understanding of the redistributive effects of commodity taxes remains elusive.

In this paper, we propose a redistributive theory of commodity taxes in monopoly screening problems. 
We adopt the perspective of a regulator---such as a state lawmaker barred from income taxation\footnote{For example, Washington state is constitutionally barred from implementing income taxes, and thus lawmakers must resort to property or commodity taxes to raise revenue and achieve redistributive goals.} or a developing country with limited state capacity\footnote{\citeauthor{besley2014taxing} suggest that in states with low capacity, commodity taxes are much easier to implement than income taxes.}---who lacks access to income taxes but can regulate the market for goods. 
We ask whether our redistribution-minded regulator can ameliorate inequality using only these a priori regressive tools. 
We show not only that the answer is yes, but also characterize (under mild regularity conditions) the least distortionary way to attain redistributive targets with minimal efficiency losses. 

\paragraph{Our Model.} Formally, we model a single monopolist interacting downstream with a continuum of consumers, each of whom has a private value $\theta$ for the good, and upstream with a tax authority (henceforth, regulator), who can levy a quantity-dependent tax on the good. We thus introduce and analyze optimal monopoly taxation in a screening model with consumer private information. 

Our objective is to find the quantity-dependent taxation schemes which maximize fairness, viewed through the lens of \emph{second order stochastic dominance} (SOSD). Say a policy $\sigma$ is more fair than another, $\sigma'$, if and only if the induced distribution of consumer surplus under $\sigma$, supposing the monopolist prices optimally, SOSD dominates the distribution induced by $\sigma'$. We say that a taxation scheme $\sigma$ is on the fairness-efficiency frontier if, among all $\sigma'$ which yield at least as much consumer surplus, $\sigma$ is SOSD-maximal. Our main contribution is to characterize the frontier of fair and efficient policies---those where consumer surplus cannot be raised without making the market less fair. 

Before previewing our results, we briefly discuss the core economic force that drives them. Because our regulator conditions only on consumers' purchases, their tax schedule separates cleanly: a uniform \emph{rebate}, which the consumer keeps in full, and a quantity-contingent \emph{tax}, which the firm  completely internalizes by adjusting its allocation schedule. These two components are linked by a budget balance requirement: the total value of the rebate must equal the revenue raised from the tax. The regulator's problem thus reduces to a single component: the design of a feasible allocation schedule $q$ (which then influences linearly the rebate consumers face). This simplification allows us to turn the taxation problem into a mechanism design problem of implementing a feasible allocation and identifying its distribution of information rents, and allows us to draw from tools in the mechanism design literature to characterize the optimal solution. 

\paragraph{Our Results.} 
To this end, we must first characterize the set of allocation schedules the monopolist can induce. Proposition \ref{p: increasing is optimal} shows that this is essentially the entire set of potentially feasible outcomes\footnote{Recall that only increasing schedules may satisfy incentive compatibility for consumers.}---for every increasing schedule, we construct a tax policy that makes it optimal for the firm to implement that schedule. Thus, even though the regulator is limited in the set of tools that are available to them, the underlying design problem is rich. 
We then leverage this result to prove our main results: a sufficiency result (Theorem \ref{t: only thresholds}) and a characterization of the frontier (Theorem \ref{t: threshold mechanisms frontier}). Theorem \ref{t: only thresholds} shows that, in order to find any maximally fair distribution of consumer surplus, it is sufficient for the regulator to restrict attention to \emph{simple excise taxes}: those which impose a per-unit tax on consumption of the good, and lump-sum redistribute the revenue to all consumers in the market. The key step is to show that any point on the fairness-efficiency frontier can be written as an optimization problem linear in the firm's choice of allocation, which allows us to apply standard extreme point results. 
Theorem \ref{t: threshold mechanisms frontier} then identifies a regularity condition on the distribution of consumers---which we call strong regularity\footnote{Formally, the probability density must be nondecreasing and log-concave.}---which is sufficient to characterize which excise taxes are on the frontier. Specifically, Theorem \ref{t: threshold mechanisms frontier} shows that the excise tax rates on the frontier are exactly those between the surplus-maximizing tax rate and the rate at the peak of the induced Laffer curve. This result is much more subtle, and relies on a novel targeting versus efficiency argument that characterizes exactly when raising revenue and redistributing it flattens the distribution of information rents. 

Together, Theorems \ref{t: only thresholds} and \ref{t: threshold mechanisms frontier} have the striking consequence that deterministic allocations are without loss of optimality: It is never strictly optimal to randomize who receives the good. Moreover, Theorem \ref{t: threshold mechanisms frontier} gives an easy way to compute which thresholds are on the frontier as a function of the marginal response of the monopolist to the tax policy. We leverage this simple characterization in an example when values are uniformly distributed to give a complete computational characterization of the set of all excise taxes that induce a distribution of consumer surplus on the frontier. 

Finally, we show under strong regularity that any mechanism on the fairness-efficiency frontier must induce strictly more rationing than the monopolist's unrestricted choice, hence only excise taxes (and not subsidies) can be on the fairness-efficiency frontier. It is thus always more efficient to attain a redistributive goal by taxing the good and lump-sum redistributing the proceeds than to subsidize the good to encourage more people to purchase (a \textit{supra-pricing} result). 
We interpret this result as identifying conditions under which commodity taxes are exactly the right mechanism to attain our regulator's redistributive goals.

\paragraph{Discussion.} Our modeling innovations allow us to take seriously the redistributive impacts of commodity taxation in a way that standard models are ill-equipped to do. In particular, absent monopoly power and private information on the side of the consumer, two seminal results constrain commodity taxes. First are the welfare theorems, which suggest redistribution and allocation should be separated, so only lump sum taxes should be considered because commodity taxes unnecessarily distort production when chasing allocative redistribution. 
Second is the result by \cite{atkinson1976design}: Given an optimal nonlinear income tax, differentiated commodity taxes are redundant as a redistributive instrument when preferences are separable between consumption and labor.

Our model departs from standard economic theory in two fundamental ways. First, we suppose a consumer's place in the distribution of surplus is both private information and endogenous to the tax policy. 
Second, we abstract away from income to isolate heterogeneity arising from consumer tastes and purchasing patterns. This shift departs from the \citeauthor{atkinson1976design} framework, where inequality is traditionally driven by heterogeneous ability and addressed via income taxation. 
Collectively, these deviations from canonical models allow us to model the specific constraints faced by redistribution-minded regulators whose primary available policy lever is commodity taxation. 

Our departures afford several advantages. First, they allow us to rationalize and evaluate the redistributive impact of commodity taxes as they are implemented in practice. Second, by explicitly excluding income taxation and direct monopoly regulation from the model, we provide rigorous second-best guidance to policymakers who lack those instruments but maintain redistributive objectives\footnote{Famously, Washington state was able to impose a capital gains tax by arguing that capital gains were not income---and therefore not property---and instead an excise tax on the sale of a financial commodity. Our analysis can shed light on what types of excise taxes Washington policymakers should use to attain progressive redistributive targets.}.

The rest of the paper is organized as follows. The remainder of this section briefly discusses the related literature. Section 2 introduces the formal model. Section 3 characterizes fair excise tax policies. Section 4 states and proves our supra-pricing property. Section 5 works through a luxury taxation example that applies our results. Section 6 concludes. All omitted proofs can be found in the appendix.

\subsection{Related Literature}

We relate to a few strands of literature. 
First, our exercise is inspired by the literature on redistributive market design (\cite{condorelli2013market}, \cite{dworczak2021redistribution}, and \cite{akbarpour2024redistributive}), which studies the impact that non-market interventions such as quotas, a public option (\cite{kang2023public}), or ordeals (\cite{tokarski2026screening}) have on allocative and redistributive efficiency. 
While we share a common motivation---seeking to understand the redistributive impact of market interventions on the distribution of consumer surplus---our tools and messages are quite distinct. In contrast to these papers, which often assume the designer's objective includes more variables than they have tools to screen on, we show that there is a role for government intervention even when consumers' welfare weights in the planner's objective are not their private information. Moreover, their policy conclusions suggest a range of possible interventions are constrained optimal (including subsidization and rationing); in contrast, we show commodity taxes can be used to attain our regulator's redistributive goals. Our redistributive criterion---which looks at the distribution of information rents in a monopoly screening problem---is similar in spirit to the criterion of \cite{pusztai2025redistributing}, though our goals and methods are quite different; for example, they study an information design problem. Closer to our paper, \cite{yang2025comparison} compare screening devices by how they distribute information rents when welfare weights are unobservable; we similarly study the distribution of information rents, but through tax-and-subsidy policies without unobservable welfare weights. 

We also relate to the literature on monopoly regulation, pioneered by \cite{baron1982regulating}, \cite{lewis1988regulating}, and more recently \cite{guo2025robust}, who study regulation of a monopolist when there is private information about firm cost or market demand (see also \cite{armstrong2007recent} for a survey). In contrast, our private information is on the side of the consumer, not the firm, and consequently our prescribed regulation tools are different. Moreover, we seek to rank entire distributions of consumer surplus, instead of maximizing weighted efficiency. While their models focus on settings such as utility regulation, we think our model, where firm marginal cost is known but consumer valuation is private, is a better model for commodity markets, or settings where the tax authority only has excise instruments and cannot directly regulate prices (i.e. optimal redistributive tariffs). More recently, this literature has also studied regulation where the regulator lacks access to the rich set of tools studied in \citeauthor{baron1982regulating}: \cite{armstrong1995nonlinear} study a model where the regulator directly constrains the prices of a monopolist screening heterogeneous consumers, while \cite{amador2022regulating} suppose transfers are not feasible and characterize the optimal price regulation. In contrast, we do not allow for price regulation at all, and instead consider a tax authority whose only instrument is to choose a tax schedule that must be measurable with respect to the observed quantity.

Further back, there is a well-developed literature studying the incidence of excise taxes and their impact on consumer surplus in public finances (see \cite{fullerton2002tax} for a survey). We will not discuss our relation to every paper in this field, but a few are particularly related. First, in the absence of screening, several papers study the impact of excise taxes on welfare in environments with a single good (\cite{anderson2001efficiency}, \cite{anderson2001tax}, \cite{carbonnier2014incidence}, and \cite{dannunzio2020multipart}). 
In contrast, we are able to study the impact of excise taxes on welfare in the presence of consumer private information, and provide a full characterization of the monopolist's optimal behavior. Second, our supra-pricing result is similar in spirit to the result that monopolists may oversupply quality relative to first best in two-sided markets (\cite{kind2008efficiency}), though we obtain it in a distinct environment.

Finally, we comment on a smaller literature that grounds redistributive motives for commodity taxes under different distortions in the market. Our notion of fairness is related to equity concerns first raised in \cite{atkinson1970measurement}. While \cite{atkinson1976design} show that in separable economies with income commodity taxes are never optimal from a redistributive standpoint, \cite{Saez2002} showed commodity taxes can be optimal in an \citeauthor{atkinson1976design} framework when tastes and earning ability are correlated; further afield,  \cite{AllcottLockwoodTaubinsky2018} and \cite{AllcottLockwoodTaubinsky2019} give a salience-based microfoundation for optimal commodity taxes, and explicitly allow the designer to have a redistributive motive.
We depart from this literature by (1) removing income taxes and (2) explicitly modeling a monopolistic seller. We contribute to it by computing, under strong regularity, the set of all frontier outcomes inducible by taxes on the fairness-efficiency frontier when the designer has a redistributive motive for consumer surplus.

\section{Model}

A single monopolist supplies a good at zero marginal cost. They interact upstream with a single regulator, and downstream with a continuum of consumers indexed by an unknown type $\theta \in \Theta$, which represents the consumers' valuation for the good. In particular, each consumer obtains $\theta q$ from consuming $q \in [0, 1]$ of the good. Suppose $\theta \sim F \in \Delta(\Theta)$, where $F$ is a commonly known full support distribution that admits a positive twice-differentiable density $f$. Moreover, suppose $F$ is Myersonian regular so that it has a monotone virtual value. 

Fix some Borel message space $M$. A \emph{game} is a measurable function $\Gamma: M \to \Delta([0, 1]) \times \bb{R}$ specifying a (potentially random) allocation quantity and transfer for each message; let $G$ be the set of all games. We assume throughout that there must always exist a message $m_0$ with $\Gamma(m_0) = (0, 0)$. When there is no ambiguity we will write a game $\Gamma$ as $(q, t) = \Gamma$, representing the quantity and transfer rules, respectively. 

The timing (and strategies) of players are as follows. 
\begin{enumerate}
    \item The government chooses a (bounded, measurable, upper semi-continuous) tax policy $\sigma: [0, 1] \to \bb{R}$ mapping firm quantities into subsidies or taxes to be given to the consumer. 
    \item Given $\sigma$, the firm chooses a game $\Gamma \in G$. 
    \item Given $\Gamma$ and $\sigma$, consumers choose a messaging strategy $m: \Theta \to \Delta(M)$.
    \item Given $m$, $\Gamma(m)$ is realized. Consumers of type $\theta$ get payoffs\footnote{Here we abuse notation, since $m(\theta)$ is a lottery over messages; we are extending the value of $q, t$ to lotteries in the usual way by taking an expectation over the realizations of randomness, in particular 
    \[ U(\theta) = \bb{E}_{m \sim m(\theta)}\bb{E}_{q \sim q(m)}[\theta q - t(m) + \sigma(q)] \]}  
    \[ U(\theta) = \theta q(m(\theta)) - t(m(\theta)) + \sigma(q(m(\theta)))\]
    while the firm gets a payoff of $\bb{E}_F[t(m(\theta))]$. 
\end{enumerate}

We remark here that we only allow the regulator to condition on the observed quantity purchased by the consumer. We impose this restriction for three reasons. First, it mirrors the implementation of commodity taxes in practice: specific excises on tobacco, alcohol, and motor fuel are assessed per unit sold. A tax authority observing a transaction sees what was bought, not who bought it or why. Second, the consumer is not allowed to privately report their type to the government. This captures an institutional constraint mirroring the administration of excise taxes by governments; they operate through the seller as a remitting agent, without communication between the government and the consumer directly about valuations. Third, because the government only allocates transfers without any good, including communication would not meaningfully increase the set of tax policies available to the principal, but would dramatically complicate the analysis.

A mechanism is a tuple $(\sigma, \Gamma, m)$ consisting of a tax policy, a game, and a messaging strategy. Two mechanisms are \emph{equivalent} if they lead to the same allocation, transfer, and subsidy rule on the path of play. A mechanism is \emph{direct} if $M = \Theta \cup \{m_0\}$ and \emph{incentive compatible} if $m(\theta) = \theta$ for each type $\theta$. More formally, incentive compatibility requires that (1) the mechanism is direct (so that $M = \Theta \cup \{m_0\}$, where $m_0$ is the outside option), and (2) 
\[\theta \in \argmax_{\theta' \in M} \left\{ \theta q(\theta') - t(\theta') + \sigma(q(\theta')) \right\}.\]

Let $\mcal M(\Gamma; \sigma)$ be the set of all messaging strategies such that each type is best responding to $\Gamma$. 
A game $\Gamma$ is \emph{optimal} given $\sigma$ if it attains 
\[ \sup_{\Gamma \in G} \sup_{m \in \mcal M(\Gamma; \sigma)} \bb{E}_F[t(m(\theta))].\]

We now have our first result, a version of the revelation principle. 

\begin{proposition}[Revelation Principle]
    \label{p: revelation principle} 
    Fix $(\sigma, \Gamma, m)$ with $m \in \mcal M(\Gamma; \sigma)$. There exists an equivalent incentive compatible direct mechanism $(\sigma, \{q, t\}, \text{Id})$. If $\Gamma$ was optimal given $\sigma$ then so is $(\{q, t\})$. 
\end{proposition}

The first part of Proposition \ref{p: revelation principle}---that it is without loss of generality to restrict to truthful incentive compatible messaging strategies for the consumers---follows from standard arguments. The novel part of the argument is that such direct mechanisms are \textit{optimal} given $\sigma$ so long as $(\Gamma, m)$ was optimal given $\sigma$ in the original mechanism. 

Given Proposition \ref{p: revelation principle}, we can define the induced consumer surplus (alternatively, consumer information rents) given a direct mechanism $(q, t)$ and tax policy $\sigma$ as 
\[ I(\theta; q, t, \sigma) = \theta q(\theta) - t(\theta) + \sigma(q(\theta)). \]

\section{The Fairness-Efficiency Frontier}

Our main results identify when and how tax policies affect the distribution of consumer surplus. This section proceeds as follows.  First, we characterize the set of all allocations a regulator can implement subject to firm optimality and incentive compatibility (Proposition \ref{p: increasing is optimal}). Second, we leverage this characterization and extreme point techniques to show threshold mechanisms are optimal regardless of the regulator's welfare weights (Theorem \ref{t: only thresholds}).  Third, we impose a regularity condition on $F$ to characterize \textit{which} thresholds are on the frontier (Theorem \ref{t: threshold mechanisms frontier}).

\subsection{Firm Optimality}

How do firms price optimally given a tax policy $\sigma$? Fix one such $\sigma$. The envelope theorem implies that 
\[ \int_{0}^{\theta} q(s) ds = \theta q(\theta) - t(\theta) + \sigma(q(\theta)) - \left( 0 q(0) - t(0) + \sigma(q(0)) \right). \]
Reading off this expression gives that the information rent generated to the lowest type is given by $\sigma(q(0)) - t(0)$; that is, the subsidy at the lowest type modifies only the individual rationality constraint.
Standard arguments imply that at an optimal mechanism, the individual rationality constraint must bind for the lowest type. Because not participating in the mechanism gives the agent an allocation $q = 0$ and transfer $t = 0$, this implies 
\[ 0 q(0) - t(0) + \sigma(q(0)) = \sigma(0). \]
Rearranging the expression above yields an expression for revenue, $\int_{\Theta} t(s) dF(s)$, given a mechanism $q$, as 
\[ \int_{\Theta} \left(\theta q(\theta) - \int_{0}^{\theta} q(s) ds + \sigma(q(\theta)) - \sigma(0) \right) dF(\theta) \]

Integrating by parts and rearranging thus gives that the firm's revenue for some fixed $q$ is exactly 
\[ R(q; \sigma) = \underbrace{\int_{\Theta} q(\theta) \left(\theta - \frac{1 - F(\theta)}{f(\theta)}  \right) f(\theta) d\theta}_{\text{Myersonian Revenue}} + \underbrace{\int_{\Theta} (\sigma(q(\theta)) - \sigma(0)) f(\theta) d\theta}_{\text{Policy Distortion}}.  \]
Standard arguments imply that $q$ is implementable by some transfer rule $t$ if and only if it is increasing; we thus have that the firm's optimal pricing problem solves 
\[ R(\sigma) = \max_{q \text{ increasing}} R(q; \sigma).\]
Because $\sigma$ is measurable with respect to $q$ and not $t$, the firm fully internalizes the quantity-contingent component $\sigma(\cdot) - \sigma(0)$, and adjusts $t$ 1-for-1.  The regulator's only remaining lever is thus changing the induced allocation $q$; once $q$ is fixed, budget balance pins the uniform component $\sigma(0)$, which accrues to consumers and is financed out of firm profit. 

Classical Myersonian reasoning tells us that all and only increasing allocations $q$ are implementable by some transfer rule for a fixed tax. The main result of this subsection is a stronger, complementary result: all (and only) increasing allocations are part of an \textit{optimal} choice for the firm for some tax policy. In particular, for any increasing $q$ we construct some tax policy $\sigma$ such that $q$ maximizes the firm's revenue. Formally,

\begin{definition}
\label{d: firm optimality}
    An allocation $q$ is \emph{optimal} given $\sigma$ if 
    \[ q \in \argmax_{\tilde q \text{ implementable}} R(\tilde q; \sigma). \]
\end{definition}

\begin{proposition}
\label{p: increasing is optimal} $q$ is nondecreasing if and only if there exists $\sigma$ where $q$ is optimal given $\sigma$. 
\end{proposition}

We interpret this result as a version of a \emph{taxation principle}: anything that is incentive compatible can be made optimal for the firm by some tax/subsidy policy. This implies that there are a priori no restrictions that can be made on the set of allocation rules which can be induced by a tax policy. Consequently, there might exist points on the fairness-efficiency frontier which are quite complicated, and which require complex taxation schemes to attain. We turn to this problem next.  

The intuition behind Proposition \ref{p: increasing is optimal} is as follows. Clearly only increasing allocations are optimal because they are the only ones that can be incentive compatible. To construct some $\sigma$ that renders $q$ optimal, first restrict to $q$ which are strictly increasing (and thus invertible). On every interior value of $q$, the first order condition is moreover sufficient by Myersonian regularity. The first-order condition defines a differential equation for $\sigma'$ in terms of the target quantity and the inverse hazard rate. Explicitly solving for this equation gives $\sigma(q)$ as the following integral equation
\[ \sigma(q(\theta)) = \sigma(0) + \int_0^{q(\theta)}  \frac{1-F(q^{-1}(s))}{f(q^{-1}(s))} - q^{-1}(s) ds. \] 
Setting $\sigma(Q)$ as above then implies $q$ is optimal for each type $\theta$, as desired. A compactness argument extends this construction to all $q$, even when it is only weakly increasing. The details are spelled out in the appendix.

One remark about the way $\sigma$ interacts with the firm's profit is worth making. Define the firm's \emph{virtual profit} from type $\theta$ to be given by 
\[ v(\theta) = \max_{Q \in [0, 1]} \left\{\psi(\theta) Q + \sigma(Q) - \sigma(0)\right\}\]
where $\psi(\theta) = \theta - H(\theta)$ is the firm's virtual value, and $H(\theta) = \frac{1 - F(\theta)}{f(\theta)}$ is the inverse hazard rate. 
Applying the envelope theorem for the firm gives that $v'(\theta) = \psi'(\theta)q(\theta)$ almost everywhere, and so we get that 
\[ v(\theta) = v(0) + \int_0^\theta \psi'(s) q(s) ds. \]
Because $\psi(\theta) q(\theta) + \sigma(q(\theta)) - \sigma(0) = v(\theta)$ for almost every $\theta$, we get that 
\[ \sigma(q(\theta)) - \sigma(0) = v(0) + \int_0^\theta \psi'(s) q(s) ds - \psi(\theta) q(\theta). \]
Integrating against $F$ and applying Fubini's theorem gives that 
\[ \int_\Theta \int_0^\theta \psi'(s) q(s) ds f(\theta) d\theta = \int_\Theta \psi'(s) q(s)\left[ \int_s^1 f(\theta) d\theta\right] ds = \int_\Theta (1 - F(s)) \psi'(s) q(s) ds\]
and hence we get that 
\[ \bb{E}_F[\sigma(q) - \sigma(0)] = v(0) + \int_\Theta \left[ (1 - F(s))\psi'(s) - f(s) \psi(s)\right]q(s) ds. \]

\subsection{Cumulative Information Rents}

This subsection gives two equivalent characterizations of our fairness-efficiency criterion. 

\begin{definition}
    An information rent $I: \Theta \to \bb{R}$ is \emph{feasible} if it is induced by an $F$-budget balanced ($\bb{E}_F[\sigma(q(\theta))] = 0$) optimal mechanism $(q, t)$ given some $\sigma$. 
\end{definition}

Feasible information rents are those that can be induced by some tax policy and optimal pricing by the firm, supposing that the government must be budget-balanced---that is, they cannot artificially add surplus into the economy\footnote{If this assumption were violated, then the government could always improve everyone's welfare by introducing an arbitrarily high lump-sum subsidy.}.

We have the following definition of the fairness-efficiency frontier for information rents. 
\begin{definition}
    \label{d: f maximality}
   Let $I, I'$ be feasible.
         \begin{itemize}
             \item $I$ is \emph{more efficient than} $I'$ if $\bb{E}_F[I] \geq \bb{E}_F[I']$. 
            \item  $I$ is \emph{more fair than} $I'$ if $(F \circ I) \succsim_{SOSD} (F \circ I')$. 
         \end{itemize}
   $I$ is \emph{$F$-unimprovable} if there does not exist a feasible $I'$ which is more fair than $I$ and more efficient than $I$, with at least one strict.\footnote{Note that it is sufficient to focus on the second order stochastic dominance part of the definition, since $I$ being more fair than $I'$ requires that $I$ is weakly more efficient than $I'$. We include the more efficient part of this definition to illuminate our frontier characterization.}
\end{definition}

Unimprovability trades off on two dimensions: efficiency, which asks about the total consumer surplus in the economy, and fairness, which tracks how that surplus is distributed among people with different types. Together, $F$-unimprovable information rents are exactly those which trace out the fairness-efficiency frontier. 

How does our definition of inequality-aversion relate to standard notions, where a designer maximizes weighted welfare with type-dependent Pareto welfare weights? 
To answer this question, we need a little bit of notation (which will be useful later in the proofs of Theorems \ref{t: only thresholds} and \ref{t: threshold mechanisms frontier} as well). 

Let 
\[ W(s) = (1 - F(s))\psi'(s) - f(s) \psi(s)\]
be the marginal transfer to the firm (see Section 3.4 for discussion). Note that if $\sigma$ is budget balanced, then 
\[ \sigma(0) = -\int_\Theta W(\theta) q(\theta) d\theta - v(0) \]
by the computations in Section 3.1, and so in particular 
\[ I(\theta; q) = \int_0^\theta q(s) ds - \int_\Theta W(s) q(s) ds - v(0) \]
is affine in $q$. Note that $I(\theta; q)$ is maximized pointwise when the lowest type is never served; we focus without loss of optimality on these mechanisms. 
Define 
\[ R_q(u) = I(F^{-1}(u); q) = \int_0^{F^{-1}(u)}q(s) ds - \int_\Theta W(s) q(s) ds\]
to be the information rents an individual of percentile $u$ receives, working with the generalized inverse whenever necessary. Define $H_q(p) = \int_0^p R_q(u) du$ to be cumulative information rents up until the $p$-th percentile. Let $G_q(x) = F(I^{-1}(x; q))$ be the distribution of information rents under allocation rule $q$ where $I^{-1}$ is the generalized inverse. 
Finally, define $ \Psi(k; p) = H_{q_k}(p) $
to be the cumulative information rents to percentile $p$ for a threshold mechanism at $k$.

With this notation, we can prove our characterization, which gives several equivalent statements about our fairness order. Let 
\[ \mcal L = \{ \lambda: \Theta \to \bb{R}_+ : \lambda \text{ is bounded, nonincreasing, and left-continuous}\}. \]

\begin{lemma}\label{l: char of SOSD by cum info rents}
    Fix any feasible information rents $I, I'$ induced by $q, q'$. The following are equivalent: 
    \begin{enumerate}
        \item $I$ is more fair than $I'$. 
        \item $H_q(p) \geq H_{q'}(p)$ for all $p \in [0, 1]$. 
        \item $\bb{E}_F[\lambda I] \geq \bb{E}_F[\lambda I']$ for all $\lambda \in \mcal L$. 
    \end{enumerate}
\end{lemma}

A few remarks about Lemma \ref{l: char of SOSD by cum info rents} are in order. First, we note the equivalence between (1) and (2) is known, \footnote{We thank Kai Hao Yang for pointing us to this reference and result therein.} and is recorded (with significantly distinct notation) in \cite{shaked2007stochastic}, Theorem 4.A.3 (albeit with an incomplete proof). For completeness of the exposition, we give a formal proof in the appendix. The key idea is to rewrite the SOSD order as an inequality based on the convex conjugates of the cumulative information rents, which follow by an integration by parts argument. 
Second, the equivalence between (1) and (3) relates our work to existing literature on redistributive market design concretely (see, for example, \cite{dworczak2021redistribution}). In particular, they fix welfare weights that are \emph{unobservable} and characterize the optimal mechanism when the designer has to screen welfare weights. In contrast, we suppose welfare weights are known and characterize the optimal mechanism to redistribute surplus when the tax authority has limited tools at their disposal (i.e. on excise taxes). Our equivalence implies that frontier mechanisms are exactly those which are on the standard Pareto frontier, but for nonincreasing (instead of arbitrary) welfare weights; i.e., our planner is explicitly inequality-averse. Intuitively, the equivalence between (1) and (2) establishes that the SOSD order cares about cumulative information rents. Given arbitrary Pareto weights, maximizing cumulative information rents becomes maximizing against cumulative Pareto weights when swapping the order of integration when computing. Then, cumulative Pareto weights are non-increasing, so the problem becomes equivalent to a designer starting out by choosing Pareto weights in $\mcal L$.

\subsection{Only Threshold Mechanisms}

We first start by introducing a particularly simple class of mechanisms; those which sell the good to all (and only) consumers above a certain cutoff type, $k$. 

\begin{definition}
    $(q, t)$ is a \emph{threshold mechanism} if $q(\theta) = \mathbf{1}_{[k, 1]}(\theta)$ for some $k \in [0, 1]$. 
\end{definition}

\begin{definition}
An \emph{excise tax} is a policy of the form $\sigma(q) = C - \tau q$ for constants $C, \tau \in \bb{R}$. 
\end{definition}

Intuitively, an excise tax subsidizes some amount $C$ for not buying ($q = 0$), and subsidizes some (generally negative) amount $C - \tau$ for buying ($q = 1$). Excise taxes are useful because they exactly characterize the set of all threshold mechanisms.

\begin{observation}
\label{p: posted prices wlog}
$(q, t)$ is a threshold mechanism if and only if $(q, t)$ is optimal given an excise tax $\sigma$. 
\end{observation}

We assume that when the firm is indifferent between multiple quantities, it breaks ties in favor of supplying the good. 
Observation \ref{p: posted prices wlog} implies that, whenever the entire fairness-efficiency frontier is traced out by threshold mechanisms, they can be implemented by exactly the set of excise taxes. Thus economies where this is the case are relatively easy to analyze. 

In principle, the set of $F$-unimprovable information rents could be induced by potentially complicated mechanisms; Proposition \ref{p: increasing is optimal} implies that any information rent which is induced by an arbitrary increasing $q$ can be implemented by some tax policy. Moreover, the second order stochastic dominance order is highly incomplete; a natural conjecture then is that the frontier of $F$-unimprovable allocations is large. Theorem \ref{t: only thresholds}, however, shows that this is not the case. In fact, we show, perhaps surprisingly, that a regulator maximizing any inequality-averse objective need only consider threshold mechanisms. Formally, 

\begin{theorem}
\label{t: only thresholds}
Fix any $\lambda \in \mcal L$ and let $I_k$ be the information rent induced by a threshold mechanism at $k$. Then 
\[ \max_{I \text{ feasible}}\ \mathbb E_F[\lambda I]\;=\;\max_{k\in[0,1]}\ \mathbb E_F[\lambda I_k]. \]
That is, any inequality-averse regulator need only search over threshold mechanisms. 
\end{theorem}

Theorem \ref{t: only thresholds} is the main takeaway of the paper. It implies that excise taxes or subsidies are enough to attain the maximum of any inequality-averse objective over the frontier, even though the tax authority could in principle implement a complicated tax schedule to induce any increasing allocation and consequently potentially rich distributions of information rents (Proposition \ref{p: increasing is optimal}). 
Consequently, simple luxury taxes or subsidies---those that impose a per-unit tax/subsidy on each unit of the item sold---(weakly) outperform any other way to redistribute surplus via commodity taxation. This includes possibly randomizing the allocation of the good (through some type of quota) or nonlinear tax schedules. 

We conclude this subsection with the proof of Theorem \ref{t: only thresholds}. 
\begin{proof}
Let $\tilde q(k)$ be the distribution over thresholds induced by an allocation rule $q$, which exists by Choquet's theorem. Note the mapping $q \to \tilde q$ is affine. Moreover note $\tilde q$ is indeed a distribution over threshold mechanisms, since all (and only) increasing allocations can possibly be implemented by some subsidy policy. 
We then have that 
\begin{align*}
     \max_{I \text{ feasible}} \bb{E}_F[\lambda I] = \max_{q \text{ nondecreasing}} \left\{ \bb{E}[\lambda(\theta) I(\theta; q) ]\right\} 
     \\ = \max_{q \text{ nondecreasing}} \left\{ \int_\Theta \lambda(\theta) \left( \int_0^\theta q(s) ds - \int_\Theta W(s) q(s) ds \right) f(\theta) d\theta \right\} 
\end{align*}
first by Proposition \ref{p: increasing is optimal} and then by our computations above. The last objective is linear in $q$, so by Bauer's maximum principle the optimum must be an extreme point of the set of nondecreasing mechanisms which by \cite{borgers2015introduction} is exactly the set of threshold mechanisms.  
\end{proof}

\subsection{Which Threshold Mechanisms?}
Theorem \ref{t: only thresholds} implies that any inequality-averse regulator need only search over threshold mechanisms in order to maximize their objective. However, it leaves open the question of \emph{which} thresholds maximize some objective for the regulator. Similarly, it leaves open the question of whether the regulator need only consider threshold mechanisms in order to induce all frontier information rents. In this subsection, we introduce a regularity condition which allows us to give a precise characterization of the set of all frontier information rents and identify which threshold mechanisms attain the frontier. 

What is the object that we need to discipline? Note that by Observation \ref{p: posted prices wlog}, the threshold mechanism with threshold $k$ is implemented by the excise tax $\tau = \psi(k)$, set equal to the virtual value at $k$; it is levied on the $1 - F(k)$ consumers who buy. The government's revenue from inducing threshold mechanism $k$ is then 
\[ A(k) = \psi(k)(1 - F(k)) = \sigma(0)\] 
where the second equality holds by budget balance. This implies that under our proposed threshold mechanism, a consumer of type $\theta$ receives 
\[ I(\theta | k) = \max\{\theta - k, 0\} + A(k); \]
differentiating in $k$ gives that the uniform rebate changes at rate 
\[ A'(k) = (1 - F(k))\psi'(k) - f(k)\psi(k) = W(k). \]
This object is then the marginal transfer to the firm. Its Laffer curve is indexed by $\tau = \psi(k)$, and so we have that $\frac{dA}{d\tau} = \frac{W(k)}{\psi'(k)}$. Hence $A'(k) = 0$ exactly when the Laffer curve is at its local extremum. 

Note that $W(\cdot)$ here is the same expression as defined in Section 3.2, where it was interpreted as the ``marginal transfer to the firm.'' To see why this interpretation is equivalent to the one above, note that 
\[ R(q; \sigma) = \int_\Theta q(\theta) \psi(\theta) f(\theta) d\theta + \int_\Theta W(s) q(s) ds \]
and hence $\int_\Theta Wq d\theta$ is firm's profit in excess of the Myersonian revenue at the same quantity. $W(\theta)$ is thus the marginal transfer to the firm necessary for the government to induce a marginally larger allocation to type $\theta$. 

The relevant regularity condition will be monotonicity of $W$, instead of monotonicity of the virtual value function. 

\begin{definition}
    $F$ is \emph{strongly regular} if $f$ is nondecreasing and log-concave. 
\end{definition}

We call this condition strong regularity because it strengthens Myersonian regularity, which is implied by log-concavity of $f$ alone. It is satisfied by some of the standard distributions in auction theory, including the uniform and power distributions ($\beta(\alpha, 1)$ for some $\alpha \geq 1$), but is not satisfied for (truncated) exponential distributions or other concave distributions which are sometimes convenient to assume. 
It yields the following regularity property of $W$. 

\begin{proposition}
\label{p: monotonicity properties of W}
Let $F \in \Delta(\Theta)$ be strongly regular. $W$ and $\frac{W}{1 - F}$ are strictly decreasing. 
\end{proposition}

The proof follows from some straightforward computations and can be found in the appendix. These monotonicity properties, however, afford us the following theorem. 

\begin{theorem}
\label{t: threshold mechanisms frontier}
    Suppose $F$ is strongly regular. Let $k_1, k_2$ solve $W(k_1) = 1-F(k_1), W(k_2) = 0$ respectively. Then,
    for any $\lambda \in \mcal L$,
    \[ \max_{I \text{ feasible}}\ \mathbb E_F[\lambda I]\;=\;\max_{k\in[k_1, k_2]}\ \mathbb E_F[\lambda I_k].\]
    Conversely, for all $k \in [k_1, k_2]$, $I_k$ is $F$-unimprovable. 
\end{theorem}

Together, we see Theorems \ref{t: only thresholds} and \ref{t: threshold mechanisms frontier} as being analogous to the standard Myersonian framework, but applied to our characterization of the fairness-efficiency frontier. First, Theorem \ref{t: only thresholds} uses an ordinal argument to extremize information rents, and shows that an optimizer can be chosen among threshold mechanisms implemented by posted prices and excise taxes; Theorem \ref{t: threshold mechanisms frontier} then uses the appropriate analogue of regularity of the virtual value function---strong regularity---to guarantee monotonicity of $W$ instead of $\psi$, and thus characterize exactly which allocations live on the frontier. 

Theorem \ref{t: threshold mechanisms frontier} has two implications. 
First, it gives a complete characterization of the set of all possible thresholds that can arise on the fairness-efficiency frontier in terms of the marginal cost of moving the monopolist ($W(\theta)$). 
Second, using our program, a redistribution-minded policymaker need not do anything fancy to hit a redistributive target subject to some minimal efficiency loss---they need only tune a one-dimensional parameter (the frontier threshold) and then implement an excise tax as desired to realize this threshold. 

What is the economic content of the threshold cutoffs $W(\theta) \geq 0$ and $W(\theta) \leq 1 - F(\theta)$? Note that aggregate consumer surplus at threshold $k$ is given by 
\[ S(k) = A(k) + \int_k^1 (1 - F(\theta)) d\theta\]
and so $S'(k) = W(k) - (1 - F(k))$. Thus the region $\{W \leq 1 - F\}$ is exactly the region where raising the threshold no longer raises aggregate consumer surplus: the rebate no longer compensates for the surplus being destroyed. Similarly, $\{W \geq 0\}$ is the region where the excise tax is on the rising side of the Laffer curve, since $A'(k) = W(k)$; Theorem \ref{t: threshold mechanisms frontier} then implies that the frontier is exactly the set of thresholds lying between the surplus-maximizing threshold and the Laffer peak.

\section{The Supra-Pricing Property}
Our analysis affords us the following comparative result. Let $(\bar q, \bar t, 0)$ be the threshold mechanism induced by the unrestricted monopoly which is subject to no taxes or subsidies. 
When are luxury taxes optimal? Note that a luxury tax must ration the good at a rate higher than the monopolist already rations, since taxes lead the monopolist to raise its effective price and sell to fewer types. A luxury tax, then, is optimal only when the redistributive effect of giving money to those who do not have the good outweighs the value of providing them with the good. It turns out that this is satisfied so long as $F$ is strongly regular. 

\begin{proposition}[Supra-Pricing]
\label{p: supra-pricing}
    Suppose $F \in \Delta(\Theta)$ is strongly regular. For any $k \in [k_1, k_2]$, $\bar q(\theta) = 0 \implies q_k(\theta) = 0$. Moreover, the monopoly allocation is not $F$-unimprovable. 
\end{proposition}

Proposition \ref{p: supra-pricing} implies that any $F$-unimprovable allocation will ration the good more than the monopolist, i.e. sell to fewer types than an unconstrained monopolist would prefer to sell to. This has two implications. First, it gives conditions under which a luxury tax is \textit{strictly} optimal from a redistributive standpoint.
To the best of our knowledge, we are the first paper to give guidance to policymakers on this point. 
Second, it gives new settings where a monopolist strictly oversupplies the good from a consumer surplus perspective. In particular, it must be that any frontier allocation (for any target level of surplus) rations the good at a rate strictly higher than the monopolist. Consequently, it can never be fairness-maximizing to redistribute the good to those with lower valuations for the good; it is always better to tax the good, raise revenue, and redistribute the resulting revenue to those who would not have bought the good in the first place. 

The intuition behind Proposition \ref{p: supra-pricing} is as follows. There are two effects that arise from increasing the marginal tax rate by a little bit. First, some people are excluded from buying the good, because it has become more expensive. Second, however, the tax raises revenue, and thus increases (by the boundary condition above) the subsidy given to everyone who does not purchase the good, $\sigma(0)$. 
While the first force makes the equilibrium distribution less fair (and efficient) in an SOSD sense---the surplus of all people weakly decreases with a tax, not accounting for the redistributive effect---the second force compensates those who no longer buy and can redistribute surplus from those with high valuations to those with lower valuations. 
Because the distribution is convex (by strong regularity),  mass is concentrated at high
valuations, so there are many consumers to tax and comparatively few to compensate. Thus taxation (as opposed to subsidization) is strictly optimal.

\section{An Example}

Consider the case now where the marginal distribution is uniform, so that $F = \text{Unif}[0, 1]$. We use this simple case to illustrate the main results of the paper. 
First, suppose there is no tax or subsidy, so the monopolist optimally posts a price of $\frac12$. What distributions of information rents are on the frontier? Note the uniform distribution is strongly regular, since $f(\theta) = 1$ is nondecreasing and log-concave. Moreover, we have that 
\[ W(\theta) = (1 - \theta) \cdot 2 - 1(2\theta - 1) = 3 - 4\theta \]
and so we know by Theorem \ref{t: threshold mechanisms frontier} that a threshold mechanism is on the frontier if and only if it has a threshold such that $W(\theta) \leq 1 - \theta$ and $W(\theta) \geq 0$. These two inequalities together 
imply that $k \in [\frac23, \frac34]$ is the set of admissible thresholds which are on the frontier, by Theorem \ref{t: threshold mechanisms frontier}. 

Note that $\frac23 > \frac12$, i.e. the set of admissible threshold mechanisms will always ration away at least $\frac16$ more of the mass of types. This is a consequence of the supra-pricing property. Hence the example is also consistent with Proposition \ref{p: supra-pricing}. Note to implement these prices, we need to set a tax $\tau = \psi(k) = 2k - 1$ for $k \in [\frac23, \frac34]$, so the excise tax ranges over $[\frac13, \frac12]$ on the frontier. 

\section{Conclusion}

This paper develops a model of redistributive commodity taxation in monopolistic markets. We use (suitably adapted) Myersonian arguments to analyze the set of all implementable and optimal allocations. 
We identify specific conditions under which frontier-optimal policies take the form of excise taxes---taxing units to induce specific posted prices. Furthermore, we establish the boundaries of the fairness-efficiency frontier, identifying when luxury taxes are strictly optimal and subsidies are excluded from the optimal policy mix. 

Our model has several limitations and opens up promising directions for future research. First is the behavior of threshold mechanisms outside of the strongly regular case. While we conjecture we can relax this assumption to the case that $W(\theta)$ is nonincreasing whenever it is nonnegative, we do not know what the weakest condition on $F$ is to ensure that we can pin down exactly which threshold mechanisms are on the frontier.
Second, characterizing when subsidies can appear on the frontier absent strong regularity is an open question.
Finally, understanding how luxury taxation interacts with income taxation seems like a particularly promising area for future research. While for this project we abstracted away from income taxation, focusing on policymakers who only have access to commodity taxes (such as Washington state), a full, more general analysis would likely endow the policymaker with both tools, and we see this as a fruitful area for future research. 

\bibliography{biblio.bib}

\newpage
\appendix
\section*{APPENDIX}
\section{Omitted Proofs}

\subsubsection*{PROOF OF PROPOSITION \ref{p: revelation principle}}
\begin{proof}

Fix $(\Gamma, m) = (\{q_\Gamma, t_\Gamma\}, m)$ and let $\{q, t\}$ be the associated direct mechanism. Clearly any deviation by type $\theta$ to sending $\theta'$ under the direct mechanism could have been replicated by sending $m(\theta')$ instead of $m(\theta)$ in the original mechanism. Thus the direct mechanism inherits incentive compatibility of the original one. 

Now we wish to show the direct mechanism defined above is optimal for the firm given a fixed $\sigma$. Suppose $(\Gamma, m)$ was optimal so 
$$\bb{E}_F[t_\Gamma(m(\theta))] \geq \bb{E}_{t_{\Gamma'}}[(m'(\theta))]$$
for every potential alternative mechanism $\Gamma' = (q', t')$ and equilibrium $m' \in \mcal M(\Gamma'; \sigma)$. Then, $(\{q, t\}, Id)$ is equivalent to $(\{q_\Gamma, t_\Gamma\}, m)$ so
$\bb{E}_F[t(\theta)] = \bb{E}_F[t_\Gamma(m(\theta))].$
Thus, 
$\bb{E}_F[t(\theta)] \geq \bb{E}_{t_{\Gamma'}}[(m'(\theta))]$
for every potential alternative mechanism $\Gamma' = (q', t')$ and equilibrium $m' \in \mcal M(\Gamma'; \sigma)$ as well. As such, $(\{q, t\}, Id)$ was optimal.
\end{proof}

\subsubsection*{PROOF OF PROPOSITION \ref{p: increasing is optimal}}
\begin{proof}
The converse. Clearly if $q$ is optimal then it is implementable and thus nondecreasing. 

The forward direction. 
    Fix an allocation $q$ and let $\Theta^o = q^{-1}((0, 1))$. We first prove the result under the assumption $q$ is strictly increasing on $\Theta^o$. 

Fix a type $\theta$; to implement $q(\theta)$, we need 
\[ q(\theta) \in \argmax_{Q \in [0, 1]} \left( \theta - \frac{1 - F(\theta)}{f(\theta)} \right) Q + (\sigma(Q) - \sigma(0))  \]
Note the objective is supermodular in $(\theta, Q)$ and so there is a selection from the maximum correspondence such that $q(\theta)$ is nondecreasing regardless of $\sigma$. 

There are now two cases. First, suppose $\theta \in \Theta^o$. In that case the first order condition is necessary, and hence we have that $q(\theta)$ must solve the first order condition 
\[ \theta - \frac{1 - F(\theta)}{f(\theta)} + \sigma'(Q) = 0. \]
We will construct $\sigma$ to satisfy the FOC  . Since $q$ is strictly monotone on $\Theta^o$, it is locally invertible, and hence setting $Q = q(\theta)$, $\theta = q^{-1}(Q)$ implies that the first order condition is solved by 
\[ \sigma'(Q) = \frac{1 - F(q^{-1}(Q))}{f(q^{-1}(Q))} - q^{-1}(Q) \]
where $Q = q(\theta)$ is the target allocation.
This formulation also allows us to verify the second order condition for optimality, which requires $\sigma''(Q) \leq 0$; in particular, note
\begin{align*}
     \sigma''(Q) = - \left( \frac{[f(q^{-1}(Q))^2 + (1 - F(q^{-1}(Q))) f'(q^{-1}(Q))] \frac{d}{dQ} q^{-1}(Q)}{f(q^{-1}(Q))^2} \right) - \frac{d}{dQ} q^{-1}(Q) 
     \\ = - \frac{d}{d\theta}\left(\theta - \frac{1 - F(\theta)}{f(\theta)} \right) \cdot \frac{d}{dQ} q^{-1}(Q) \leq 0
\end{align*}
where the last inequality follows from regularity of $F$ and the fact $q^{-1}$ is strictly increasing on $(0, 1)$ and hence $\frac{d}{dQ} q^{-1}(Q) > 0$. 
Note this formula holds for all $Q \in (0, 1)$ chosen; in particular, this implies that one solution to the differential equation for $\sigma'$ is 
\[ \sigma(Q) = \int_0^Q \frac{1-F(q^{-1}(s))}{f(q^{-1}(s))} - q^{-1}(s) ds. \]
Now suppose $\theta \not\in \Theta^o$, so that the target allocation 
is either $1$ or $0$. Define $\sigma(1)$ to be 
\[ \int_0^1 \frac{1-F(q^{-1}(s))}{f(q^{-1}(s))} - q^{-1}(s) ds \]
and let $\sigma(0)$ be chosen by continuity so that $\sigma(0) = \lims_{Q \to 0} \sigma(Q)$\footnote{We note that as constructed that $q^{-1}(s)$ is defined only on the image of $q$; we extend the inverse to the entire domain by picking an arbitrary differentiable extension that maintains optimality of the chosen quantity.}. 

Having defined $\sigma$ as above, we now need to show that $q$ is chosen optimally in response to $\sigma$. Let the firm's optimal allocation rule be $\hat{q}(\theta).$ If $\hat{q}(\theta) \in (0, 1)$ then by the FOC analysis above, $\hat{q}(\theta) = q(\theta)$. By the monotone comparative statics from earlier, $\hat{q}$ is (weakly) increasing in $\theta$. Thus, $\hat{q}^{-1}((0, 1))$ is connected; furthermore, $\hat{q} = 0$ for all $\theta$ below $\hat{q}^{-1}((0, 1))$ and $\hat{q} = 1$ for all $\theta$ above $\hat{q}^{-1}((0, 1))$. This exactly coincides with $q$ since $q$ is also weakly increasing.

We can now generalize the result to arbitrary $q(\theta)$, even if it is not strictly increasing on $\Theta^o$ via a standard upper hemicontinuity argument. In particular, note that 
\[ q \in \argmax_{\tilde q \text{ increasing}} R(\tilde q; \sigma)\]
is a compact maximization problem over a continuous domain (see \cite{borgers2015introduction} for the argument), and hence $q$ is upper hemi-continuous in $\sigma$ by Berge's maximum theorem. 

Let $\{q\}$ be some strictly increasing sequence of functions approximating pointwise a nondecreasing $q$. Next, given $\{q\}$ define $\{\sigma\}$ as above to ensure 
\[ q \in \argmax_{\tilde q \text{ increasing}} R(\tilde q; \sigma).\] 
Note by definition that $\{\sigma\}$ are uniformly bounded; moreover it is of uniformly bounded variation because the inverse hazard rate is bounded uniformly on $\Theta$ (since it is a continuous function over a compact set). Helly's generalized selection theorem (for functions of bounded variation) then implies that $\{\sigma_n\}$ has a convergent subsequence; similarly Helly's selection theorem (for monotone functions) gives a convergent subsequence $\{q_n\}$. Since $R(\tilde q; \sigma_n) \to R(\tilde q; \sigma)$ by the dominated convergence theorem, we can conclude via upper hemicontinuity that $q \in \argmax_{\tilde q \text{ increasing}} R(\tilde q; \sigma)$, where we restrict to a desired subsequence as necessary so that ${q_{n_k}} \to q$ as $\sigma_{n_k} \to \sigma$ for some tax $\sigma$. 

Finally, note that we get budget balance for free: the firm's problem depends on $\sigma$ only through $\sigma(\cdot) - \sigma(0)$, so $\sigma$ can be shifted by a constant without affecting optimality or $v(0)$. Thus there is a budget-balanced $\sigma$ that induces the desired allocation. 
\end{proof}

\subsubsection*{PROOF OF LEMMA \ref{l: char of SOSD by cum info rents}}

    \begin{proof}

    First, we show (1) and (2) are equivalent. We begin with some computations for an arbitrary CDF $G$. 
        Applying integration by parts gives
        $$\int_{-\infty}^{z} G(x) dx = [x G(x)]_{-\infty}^{z} - \int_{-\infty}^{z} x dG(x).$$
        Since all possible information rents are bounded, $[x G(x)]_{-\infty}^{z} = zG(z)$. Next, using change of variables with $x = R(u)$ and letting $p^* = G(z)$ gives $dG(x) = dG(R(u)) = du$ and $R^{-1}(-\infty) = G(-\infty) = 0, R^{-1}(x) = G(x) = p^*$ so
        $$\int_{-\infty}^{z} x dG(x) = \int_{0}^{p^*} R(u) du = H(p^*)$$
        so 
        $$\int_{-\infty}^{z} G(x) dx = z G(z) - H(p^*) =zp^* - H(p^*).$$
        This value of $p^*$ satisfies
        $$\frac{d}{dp} \left[zp - H(p)\right] = z - \frac{d}{dp} \int_{0}^{p} R(u) du = z - R(p^*) = 0$$
        and
        $$\frac{d^2}{dp^2} \left[zp - H(p)\right] = -R'(p^*) \leq 0$$
        so $p^*$ maximizes $zp - H(p)$. 
        Note that $R$ is nondecreasing, since $I'(\theta; q) = q(\theta) \geq 0$ by the envelope theorem. 
        As such, 
        $$\int_{-\infty}^{z} G(x) dx = zp^* - H(p^*) = \max_p \{zp - H(p)\}.$$
        Then, by the definition of SOSD, $F \circ I(\cdot; q) \succsim_{SOSD} F \circ I(\cdot; q')$ if and only if 
        $$\int_{-\infty}^{z} G_q(x) dx \leq \int_{-\infty}^{z} G_{q'}(x) dx \text{ for all } z.$$
        Plugging in expressions from before, this is equivalent to 
        $$\max_p \{zp - H_q(p)\} \leq \max_p \{zp - H_{q'}(p)\} \text{ for all } z.$$
        If $H_q(p) \geq H_{q'}(p)$ for all $p$ then $zp - H_q(p) \leq zp -H_{q'}(p)$ for all $p$ so taking maximums over $p$ on both sides preserves the inequality.
        Conversely, suppose for the sake of contradiction that the SOSD condition holds so that $\int_{-\infty}^{z} G_{q}(x) d x \leq \int_{-\infty}^{z} G_{q^{\prime}}(x) d x$ for all $z$, but there exists some $p^{*}$ for which $H_{q}(p^{*}) < H_{q^{\prime}}(p^{*})$. 
        Since $H_{q'}$ is convex, there exists a subgradient $z^*$ of $H_{q'}$ at $p^*$. By the definition of the convex conjugate, one has that 
        \[ \int_{-\infty}^{z^*} G_{q'}(x) dx = \max_p \left\{ z^* p - H_{q'}(p) \right\} = z^*p^* - H_{q'}(p^*).\]
        Maximizing over $p$ yields a weakly greater value than evaluating at $p^*$ specifically for the allocation $q$, and hence we get that 
        \[  \int_{-\infty}^{z^*} G_{q}(x) dx = \max_p \left\{ z^* p - H_{q}(p) \right\} \geq z^*p^* - H_{q}(p^*)\]
    Since $H_q(p^*) < H_{q'}(p^*)$, we get that 
    \[ z^* p^* - H_q(p^*) > z^* p^* - H_{q'}(p^*).\]
    Together, this implies 
    \[ \int_{-\infty}^{z^*} G_q(x) dx \geq z^* p^* - H_q(p^*) > z^* p^* - H_{q'}(p^*) = \int_{-\infty}^{z^*} G_{q'}(x) dx. \]
    Thus, 
    \[ \int_{-\infty}^{z^*} G_q(x) dx > \int_{-\infty}^{z^*} G_{q'}(x) dx\]
    contradicting the assumption $q \succsim_{SOSD} q'$. 

    Next, we give the equivalence between (2) and (3). Fix any $\lambda \in \mcal L$ and define $\mu$ such that $\mu([\theta, 1]) = \lambda(\theta)$ for all $\theta$; such a $\mu$ exists because $\lambda$ is bounded and nonincreasing. Note that 
    \[ R_q(F(\theta)) = I(\theta; q) \implies \int_0^x I(\theta; q) dF(\theta) = H_q(F(x)).\]
    This implies that 
    \[ \bb{E}_F[\lambda I] = \int_\Theta I(\theta; q) \mu([\theta, 1]) dF(\theta) = \int_\Theta\left[\int_0^x I(\theta; q) d F(\theta) \right] \mu(dx) = \int_\Theta H_q(F(x))\mu(dx) \]
    where we interchange the order of integration by Fubini's theorem. 
    Now suppose (3) holds. For fixed $x$ let $\lambda = \mathbf{1}\{\theta \leq x\} \in \mcal L$ (corresponding to $\mu = \delta_x$). We then get that $H_q(F(x)) \geq H_{q'}(F(x))$. Varying across $x$ and noting $F(x)$ is full support gives (2). For the converse, integrate the inequality $H_q(F(x)) \geq H_{q'}(F(x))$ against $\theta$  and the positive measure $\mu$ from above yields (3). 
\end{proof}

\subsubsection*{PROOF OF OBSERVATION \ref{p: posted prices wlog}}
\begin{proof}
    Let $\sigma(q) = C - \tau q$ be an excise tax. Then for $q(\theta)$ to be optimal, it must be that 
    \[ q(\theta) \in \argmax_{Q \in [0, 1]} \left( \theta - \frac{1 - F(\theta)}{f(\theta)} - \tau\right) Q + C. \]
  Yet it is clear the solution to the above is given by 
  \[ q(\theta) = \mathbf{1}\left\{ \theta - \frac{1 - F(\theta)}{f(\theta)} \geq \tau \right\}\]
  which is clearly a threshold mechanism by Myersonian regularity. This gives the converse. For the forward direction, note that we can recover every threshold mechanism by varying $\tau$ as the virtual value function is continuous and monotone. 
\end{proof}

\subsubsection*{PROOF OF PROPOSITION \ref{p: monotonicity properties of W}}
\begin{proof}
First, we remark strong regularity implies convexity of $F$ and concavity of the virtual value function. To see this, note that 
    \[ \frac{d}{ds} \log(f) = \frac{f'}{f} \geq 0 \iff f' \geq 0 \]
    which gives convexity of $F$. 
    Second, let $H(s) = \frac{1 - F(s)}{f(s)}$ denote the inverse hazard rate. Note that 
    \[ \psi(s) = s - H(s) \implies \psi''(s) = -H''(s)\]
    so we need that $H''(s) \geq 0$ for $\psi$ to be concave. A computation gives that 
    \[ H'(s) = -1 - \frac{(1 - F(s))f'(s)}{f(s)^2}  = -1 - H(s) \frac{f'(s)}{f(s)} \]
    Differentiating again gives 
    \[ H''(s) = -\left( H'(s)\frac{f'(s)}{f(s)} + H(s)\left( \frac{f(s) f''(s) - f'(s)^2}{f^2} \right) \right) \]
    Note moreover that $\frac{d}{ds} \log(f(s)) \geq 0$, and
      \[ \frac{d^2}{ds^2}\log(f(s)) = \frac{d}{ds} \frac{f'(s)}{f(s)} = \frac{f(s) f''(s) - f'(s)^2}{f(s)^2} \leq 0\]
      by log-concavity of $f$.
      Hence $H''(s)$ is nonnegative, since 
    \begin{enumerate}
        \item $H'(s) \leq 0$ by strong regularity (the inverse hazard rate is decreasing). 
        \item $\frac{f'(s)}{f(s)} \geq 0$ since $f' \geq 0$ and $f > 0$. 
        \item $H(s) \geq 0$ by definition. 
        \item $\frac{f(s)f''(s) - f'(s)^2}{f(s)^2} \leq 0$ by log-concavity of $f$. 
    \end{enumerate}
  Hence $H''(s) \geq 0$, so $\psi''(s) \leq 0$ and hence $\psi$ is concave. 
  
Now nonincreasingness of $W$.
A basic computation implies we can rewrite $W(\theta)$ as 
\[ W(\theta) = 3(1-F(\theta)) - \theta f(\theta) + \frac{(1-F(\theta))^2 f'(\theta)}{f(\theta)^2}\]
and 
\[ W'(\theta) = -4f(\theta) - \theta f'(\theta) + \frac{d}{d\theta} \left[ \frac{f'(\theta)}{f(\theta)} \cdot \frac{(1-F(\theta))^2}{f(\theta)} \right]\]
by a computation. 

\noindent Since strong regularity implies $f' \geq 0$, we know that 
\[ \frac{d}{d\theta}\frac{(1-F(\theta))^2}{f(\theta)}=  -\frac{1-F(\theta)}{f(\theta)^2} [2f(\theta)^2 + (1-F(\theta))f'(\theta)] \leq 0 \text{  and  } -\theta f'(\theta) \leq 0.\]
Moreover, strong regularity implies $\frac{f'(\theta)}{f(\theta)}$ is nonnegative and nonincreasing, and hence we know that 
\[ \frac{d}{d\theta} \left[ \frac{f'(\theta)}{f(\theta)} \cdot \frac{(1-F(\theta))^2}{f(\theta)} \right] \leq 0. \]
Together, these two observations allow us to bound $W'(\theta)$ from above by $-4 f(\theta)$. Since $f$ is fully supported, $W'(\theta) < 0$, as desired. 

The second statement. Note
    \[ \frac{W(s)}{1 - F(s)} = \psi'(s) - \frac{f(s)}{1 - F(s)} \psi(s) = \psi'(s) - \frac{\psi(s)}{H(s)} \]
    where $H(s) = \frac{1 - F(s)}{f(s)}$. 

    Further simplifying gives that this can be written as 
    \[ \psi'(s) - \frac{\psi(s)}{H(s)} = \psi'(s) - \frac{s - H(s)}{H(s)} = \psi'(s) - \frac{s}{H(s)} + 1 \]
    and so its derivative can be written as 
    \[ \psi''(s) - \frac{H(s) - H'(s) s}{H(s)^2} \]
    Recall that $\psi''(s) \leq 0$. Moreover, note $H(s) > 0$ always on the interior, while by strong regularity $H'(s) \leq 0$, i.e. the inverse hazard rate is decreasing whenever it is defined.  Thus the second $\frac{H(s) - H'(s) s}{H(s)^2}$ term is the sum of nonnegative terms and hence is strictly positive, but is being subtracted. Thus, $\frac{W(s)}{1 - F(s)}$ is strictly decreasing. 
\end{proof}

\subsubsection*{PROOF OF THEOREM \ref{t: threshold mechanisms frontier}}
\begin{proof}
First, the converse. We show that if $(q, t, \sigma)$ is a threshold mechanism inducing a threshold in the interval, then it must be on the frontier. 

\begin{lemma}
\label{l: all thresholds on frontier}
   Fix a tax policy $\sigma$ inducing a threshold mechanism with threshold $k \in \{\theta : W(\theta) \leq 1 - F(\theta), W(\theta) \geq 0\}$, and let $I$ be its information rent. Then $I$ is $F$-unimprovable. 
\end{lemma}
\begin{proof}
Recall that if $I(\theta | k)$ represents the information rent of a consumer of type $\theta$ given threshold $k$, then
\[ I(\theta | k) = \max\left\{\theta - k, 0\right\} - \int W(s) \mathbf{1}_{\{s \geq k\}} ds. \]
 Let $I_p(k)$ denote the $p$-th percentile of aggregate information rents. Cumulative information rents up to percentile $p$ induced by thresholds $k$ are then 
    \[ H(p | k) = \Psi(k; p) = \int_0^p R(u) du. \]
    Differentiating gives that 
    \[ \frac{\partial H(p | k)}{\partial k} = \bb{E}_F\left[ \frac{\partial I(\theta | k)}{\partial k}\mathbf{1}_{\{I(\theta | k) \leq I_p(k)\}}  \right]\]
   using the Leibniz integral rule as necessary. 
A further computation gives
    \[ \frac{\partial I(\theta|k)}{\partial k} = -\mathbf{1}_{\{\theta \geq k\}} + W(k).\]
Combining this with the above expression gives 
    \[ \frac{\partial H(p | k)}{\partial k} = \bb{E}\left[ \left( - \mathbf{1}_{\{\theta \geq k\}} + W(k) \right) \mathbf{1}_{I(\theta | k) \leq I_p(k)} \right] = p W(k) - \bb{E}\left[\mathbf{1}_{\{\theta \geq k\}}\mathbf{1}_{I(\theta | k) \leq I_p(k)} \right].\]
    Consider now the set 
    \[ k \in \{\theta: W(\theta) \leq 1 - F(\theta), W(\theta) \geq 0\}.\]
    Since we know that 
    \[ W(0) = 3(1 - F(0)) - 0f(0) + \frac{(1 - F(0))^2f'(0)}{f(0)^2} = 3 + \frac{f'(0)}{f(0)^2} > 1 \]
    we have that $k_1 > 0$. This implies that no agent in the set $\Theta_0 = [0, k_1)$ consumes any good at any potential threshold. Since $F$ is fully supported, $F(k_1) > 0$. 
    Taking $p  = \frac{F(k_1)}{2}$ gives that 
    \[ \frac{\partial H(p | k)}{\partial k} = \frac{F(k_1)}{2} W(k) - 0\]
    since no agent below the $p$-th percentile receives any good. Since $W(k) \geq 0$ for all potential thresholds, this gives that $\frac{\partial H(p | k)}{\partial k}$ evaluated at $p = F(k_1)/2$ is nonnegative for all potential thresholds $k$. 

    Finally, taking $p = 1$ instead gives that 
    \[ \frac{\partial H(1 | k)}{\partial k} = W(k) - \bb{E}[\mathbf{1}_{\left\{\theta \geq k\right\}}] = W(k) - (1 - F(k)). \]
    Since $W(k) \leq 1 - F(k)$ for every potential threshold, this gives that $\frac{\partial H(1|k)}{\partial k} \leq 0$. 

    Together, these two inequalities imply for any potential threshold $k$ on the frontier, raising $k$ will always increase cumulative information rents for low $p$ and decrease them for high $p$, and so no threshold within the set of potential thresholds can dominate another. 

    Finally, we show that no mixture of thresholds can dominate any threshold inside the desired interval.
    To do this, we construct some nonincreasing welfare weight $\lambda$ such that the threshold at $k$ globally maximizes $\bb{E}_F[\lambda I]$ across all pure thresholds; linearity of the objective in $q$ then implies that it must maximize the objective over all convex combinations of thresholds, giving the desired argument. 

    Fix some $k$ where $W(k) \in [0, 1 - F(k)]$. Construct our nonincreasing welfare weight in the following way. Let our regulator place weight $\alpha$ on cumulative rents at percentile $p_L \leq F(k_1)$, and the remainder of the weight on cumulative rents at $p_H = 1$. The weighted objective is then 
    \[ J(x) = \alpha H(p_L | x) + (1 - \alpha)H(1 | x);\]
    a computation gives  
    \[ J'(x) = \alpha p_L W(x) + (1-\alpha) \big[ W(x) - (1 - F(x)) \big]. \]
    Choosing 
    \[ \alpha = \frac{1 - F(k) - W(k)}{1 - F(k) - W(k) + p_L W(k)} \]
    implies that $J'(k) = 0$, so $k$ maximizes this objective. Note $\alpha \in [0, 1]$ since $W(k) > 0$ and $1 - F(k) - W(k) \geq 0$. 

    It suffices to consider $x\geq k_1$: by the computation in the forward direction
below, $\frac{\partial\Psi(x;p)}{\partial x}>0$ for every $x<k_1$ and every $p$, so
$\Psi(x;p)<\Psi(k_1;p)$ for all $p>0$ and hence $J(x)<J(k_1)$ by Lemma \ref{l: char of SOSD by cum info rents}. For
$x\ge k_1$ we have $F(x)\ge F(k_1)\ge p_L$, so the displayed $J'$ is correct.

    To show this is a global maximum, note that 
    \[ \frac{J'(x)}{1 - F(x)} = \big[ \alpha p_L + 1 - \alpha \big] \frac{W(x)}{1-F(x)} - (1-\alpha).\]
    Since $\frac{W(x)}{1 - F(x)}$ is strictly decreasing in $x$, $\frac{J'(x)}{1 - F(x)}$ is strictly decreasing in $x$ and zeros at $x = k$. This implies that $J(x)$ is strictly single-peaked, so $k$ is the unique global maximizer among all pure thresholds. Thus $k$ maximizes the desired objective over all feasible allocations, so by Lemma \ref{l: char of SOSD by cum info rents} no feasible allocation can strictly dominate it at every percentile. 
\end{proof}

Now we establish the forward direction that it is sufficient for our regulator to restrict to threshold mechanisms when seeking to maximize their objective. 
 
Fix some nonincreasing $\lambda$. Theorem \ref{t: only thresholds} then implies that there is some threshold mechanism with threshold $k$ that attains this maximum. 
We now show that $k$ must be in our candidate interval. Note 
$$\Psi(k; p) = \int_0^{p} [\max\{F^{-1}(t)-k, 0\} - \int_k^1 W(s)ds] dt$$
so $$\frac{d\Psi}{dk} = pW(k) - \max\{p - F(k), 0\}.$$
Next, let $k_1$ solve $W(k_1) = 1-F(k_1)$ and $k_2$ solve $W(k_2) = 0$ so the set of admissible thresholds is $[k_1, k_2]$. If $k < k_1$ then $W(k) > 1-F(k)$. If $p < F(k)$ then
$$\frac{d\Psi}{dk} = pW(k) > 0$$
while if $p \geq F(k)$ then
$$\frac{d\Psi}{dk} > p(1-F(k)) - p + F(k) = (1-p)F(k) \geq 0$$
so $\frac{d\Psi}{dk} > 0$ for all $k < k_1$ and all $p$. As such, any threshold $k < k_1$ is dominated by moving that up to a threshold exactly at $k_1$.

On the other hand, if $k > k_2$ then $W(k) < 0$ so $\frac{d\Psi}{dk} < 0$ for all $p > 0$. As such, any threshold mechanism at a threshold $k > k_2$ is dominated by moving that threshold down to $k_2$. This finishes the argument. \end{proof}

\subsubsection*{PROOF OF PROPOSITION \ref{p: supra-pricing}}
\begin{proof}
We start with a computation. 
    \begin{align*}
        W(s) = (1 - F(s))\psi'(s) - f(s)\psi(s) = 3(1 - F(s)) - s f(s) + \frac{(1 - F(s))^2}{f(s)^2} f'(s) \\ = (1 - F(s)) - s f(s) + 2 (1 - F(s)) + \frac{(1 - F(s))^2}{f(s)^2} f'(s) 
    \end{align*}
    At the monopoly price, we know that $\theta^* = \frac{1 - F(\theta^*)}{f(\theta^*)}$ and so 
    \[ 1 - F(\theta^*) - \theta^* f(\theta^*) = f(\theta^*) \left(\frac{1 - F(\theta^*)}{f(\theta^*)} - \theta^* \right) = 0\]
    and so we get that 
    \[ W(\theta^*) = 2(1 - F(\theta^*)) + \frac{(1 - F(\theta^*))^2}{f(\theta^*)^2} f'(\theta^*) > 1 - F(\theta^*). \]
    In particular, this implies that at the monopoly threshold, $W(\theta^*) > 1 - F(\theta^*)$. By strong regularity and Theorem \ref{t: threshold mechanisms frontier}, $\theta^*$ cannot be a threshold which is on the fairness-efficiency frontier. 
    
    Finally, consider $1 - \frac{W(\theta)}{1 - F(\theta)}$. A necessary condition by Theorem \ref{t: threshold mechanisms frontier} for this to be on the frontier is that this expression is nonnegative. However, we know that $1 - \frac{W(\theta^*)}{1 - F(\theta^*)} < 0$; since $\frac{W(\theta)}{1 - F(\theta)}$ is strictly decreasing by Proposition \ref{p: monotonicity properties of W}, for this to be positive it must be that any admissible threshold has $k > \theta^*$. This finishes the proof. 
\end{proof}

\end{document}